\documentclass[11pt,a4paperwide]{amsart}
\usepackage[left=3.5cm, top=3cm, bottom=3cm, right=3.5cm]{geometry}
\usepackage{amsmath,amsfonts,amssymb,amsthm,amscd}

\usepackage{amsthm}
\theoremstyle{plain}

\newtheorem*{theoremA*}{Theorem A}
\newtheorem*{theoremB*}{Theorem B}

\usepackage{graphicx}
\usepackage{epsfig}

\makeatletter
\newcommand\xleftrightarrow[2][]{%
  \ext@arrow 9999{\longleftrightarrowfill@}{#1}{#2}}
\newcommand\longleftrightarrowfill@{%
  \arrowfill@\leftarrow\relbar\rightarrow}
\makeatother

\newtheorem{thm}{Theorem}
\newtheorem{lem}[thm]{Lemma}
\newtheorem{prop}[thm]{Proposition}
\newtheorem{defn}[thm]{Definition}
\newtheorem{cor}[thm]{Corollary}

\newtheorem{example}[thm]{Example}

\def\circa{\overset{\alpha}{\circ}}
\def\circb{\overset{\beta}{\circ}}

\def\deltaa{\delta_\alpha}
\def\deltab{\delta_\beta}

\def\circone{\overset{1}{\circ}}
\def\circtwo{\overset{2}{\circ}}
\def\circk{\overset{k}{\circ}}

\def\circequalsone{\overset{i=1}{\circ}}

\begin{document}

\title[A construction of Multidimensional Dubrovin-Novikov Brackets]{A construction of Multidimensional Dubrovin-Novikov Brackets}
\date{\today}

\author{Ian A.B. Strachan}

\address{School of Mathematics and Statistics\\ University of Glasgow\\Glasgow G12 8QQ\\ U.K.}
\email{ian.strachan@glasgow.ac.uk}

\keywords{Hamiltonian structures, integrable systems} \subjclass{}

\begin{abstract}
A method for the construction of classes of {\sl examples} of multi-dimensional, multi-component Dubrovin-Novikov brackets of hydrodynamic type is given. This is based on an extension of the original construction of Gelfand and Dorfman which gave examples of Novikov algebras in terms of structures defined from commutative, associative algebras. Given such an algebra, the construction involves only linear algebra. 
\end{abstract}

\maketitle

\tableofcontents

\section{Introduction}

There has been a recent renaissance in the study of multi-component, multi-dimensional Poisson brackets of Dubrovin-Novikov type. These brackets, in $D$-dimensions and with $N$-components, have the form
\begin{equation}\label{DNbracket}
\{I,J\} = \sum_{\alpha=1}^D \int \frac{\delta I}{\delta u^i} \mathcal{H}^{ij} \frac{\delta J}{\delta u^j} \, d^DX\,,
\end{equation}
where the Hamiltonian operator $\mathcal{H}^{ij\alpha}$ is a local, first-order operator of the form
\[
\mathcal{H}^{ij\alpha} = g^{ij\alpha}({\bf u}) \frac{\partial~}{\partial X^\alpha} + b^{ij\alpha}_k({\bf u}) u^k_{X^\alpha}\,.
\]
The study of such brackets were initiated by Dubrovin and Novikov in their seminal papers \cite{DN} and \cite{DN2}, and the full conditions (to ensure that the bracket is skew and satisfies the Jacobi identity)
were derived by Mokhov \cite{M}  and are given in Theorem \ref{mokhov1} below. However, the classification of such brackets remains an open problem.

\medskip

The most general statement, combining results from \cite{M} and \cite{M2} may be stated thus:

\begin{thm}
Consider a non-degenerate multi-dimensional, multi-component bracket.

\begin{itemize}

\item[$\bullet$] For $D=1\,,$ such a bracket may be reduced, by a change of variable, to a constant, or Darboux, form;

\item[$\bullet$] For $D\geq 2\,,$ such a bracket may be reduced, by a change of variable, to a linear form where
\[
\mathcal{H}^{ij\alpha}({\bf u}) = \left\{  \left( b^{ij\alpha}_k + b^{ji\alpha}_k\right) u^k \frac{d~}{d X^\alpha} + b^{ij\alpha} u^k_{X^\alpha} \right\} + \left\{ g^{ij\alpha}_o \frac{d~}{dX^\alpha} \right\}\,.
\]
Furthermore, tensorial conditions exist to determine whether a full reduction to constant form, where all components of all the metrics are constant,  may be achieved.

\end{itemize}

\end{thm}

\noindent Without the non-degeneracy condition, much less is known.

\medskip

On substituting this linear structure into the defining condition in Theorem \ref{mokhov1} one obtains equations for the constants $b^{ij\alpha}_k$ and
$g^{ij\alpha}_o\,.$ These are best described in terms of an algebraic structure on regarding these constants are structure constants of a multiplication and inner produce, i.e.
\begin{eqnarray*}
e^i \circa e^j & = & b^{ij\alpha}_k e^k\,,\\
\langle e^i, e^j \rangle_\alpha & = & g^{ij\alpha}_o\,.
\end{eqnarray*}
These algebraic conditions were defined implicitly in \cite{M} and will be presented explicitly below. In the $D=1$ case they define a Novikov algebra. For $D\geq 2$ one has multiple Novikov algebras connected by crucial compatibility conditions. Finding a construction that generates {\sl examples} of multiplications satisfying these compatibility conditions is the raison d'\^etre of this article.

\medskip

The renaissance in the study of such brackets has come from two directions. Progress has been made in the classification of these brackets, including in the case where the structures are degenerate. Firstly \cite{FLS,Sav}:

\begin{itemize}

\item{} For $D=2$, a classification of non-degenerate and degenerate has been completed in the cases when the number of variables is small, with partial results for multi-component systems;

\item{} Full results for $D=3\,,N=3$ have been obtained, extending the results of Mokhov who studied the $D=3\,,N=1\,,2\,.$

\end{itemize}

\noindent Secondly, by developing the theory of Poisson Vertex Algebras, a deformation theory has been developed \cite{C}. Using the classification results for $D=N=2\,,$ the various cohomology group that govern such deformations have been constructed.

\medskip

However, there is currently no classification (in the sense of an explicit, exhaustive, list of examples) of such brackets. This is true even for linear Poisson brackets in the $D=1$ case: Novikov algebras have only been classified in the cases where the number of components is small, $N\leq 3\,$ \cite{BM1,BM1b,BG}. This lack of examples also hampers the programme of constructing a deformation theory: one requires something to deform.

\medskip

The aim of this paper is to provide explicit examples of multi-component, multi-dimensional Dubrovin-Novikov brackets. The approach is far from exhaustive - it is certainly not a step in any future classification programme, but it does provide a mechanism for the construction of examples using only linear algebra. The approach is based on a recent result in the study and classification of Novikov algebras \cite{Fil,Xu}.

\section{Multidimensional Dubrovin-Novikov brackets}

The full conditions for the bracket (\ref{DNbracket}) to define a Poisson bracket places very strong conditions on the functions $g^{ij\alpha}({\bf u})$ and $b^{ij\alpha}_k({\bf u})$ (see \cite{M,M2}):

\begin{thm}\label{mokhov1}
A bracket (\ref{DNbracket}) is a Poisson bracket if and only if the following relations for the coefficients of the bracket are satisfied:
\begin{eqnarray*}
g^{ij\alpha} &=& g^{ji\alpha}\,,\\
\frac{\partial g^{ij\alpha}}{\partial u^k} &=& b^{ij\alpha}_k + b^{ji\alpha}_k\,,\\
\sum_{(\alpha,\beta)} \left( g^{si\alpha} b^{jr\beta}_s\right. &-& \left.g^{sj\beta} b^{ir\alpha}_s\right)  =  0 \,,\\
\sum_{(i,j,k)} \left( g^{si\alpha}b^{jr\beta}_s\right. &-&\left. g^{sj\beta} b^{ir\alpha}_s\right)  =  0 \,,\\
\sum_{(\alpha,\beta)}  g^{si\alpha}\left( \frac{\partial b^{jr\beta}_s}{\partial u^q} - \frac{\partial b^{jr\beta}_q}{\partial u^s} \right)&+&
b^{ij\alpha_s} b^{sr \beta}_q - b^{ir\alpha}_s b^{sj\beta}_q   =  0 \,,\\
g^{si\beta} \frac{\partial b^{jr\alpha}}{\partial u^s} - b^{ij\beta}_s b^{sr\alpha}_q - b^{ir\beta}_s b^{js\alpha}_q & = &g^{sj\alpha} \frac{\partial b^{ir\beta}}{\partial u^s} - b^{ji\alpha}_s b^{sr\beta}_q - b^{jr\alpha}_s b^{is\beta}_q\,,\\
\frac{\partial~}{\partial u^k} \left[
g^{si\alpha}\left( \frac{\partial b^{jr\beta}_s}{\partial u^q} - \frac{\partial b^{jr\beta}_q}{\partial u^s} \right)
b^{ij\alpha}_s b^{sr \beta}_q - b^{ir\alpha}_s b^{sj\beta}_q \right]&+&\sum_{(i,j,r)} \left[ b^{si\beta}_q \left( \frac{\partial b^{jr\alpha}_k}{\partial u^s} - \frac{\partial b^{jr\alpha}_s}{\partial u^k}\right)\right]\\
+
\frac{\partial~}{\partial u^q} \left[
g^{si\beta}\left( \frac{\partial b^{jr\alpha}_s}{\partial u^k} - \frac{\partial b^{jr\alpha}_k}{\partial u^s} \right)
b^{ij\beta}_s b^{sr \alpha}_k - b^{ir\beta}_s b^{sj\alpha}_k \right]&+&\sum_{(i,j,r)} \left[ b^{si\alpha}_k \left( \frac{\partial b^{jr\beta}_q}{\partial u^s} - \frac{\partial b^{jr\beta}_s}{\partial u^q}\right)\right]=0\,.
\end{eqnarray*}

\end{thm}

\noindent The following Lemma, again from \cite{M,M2}, follows from the analysis of the above conditions.

\begin{lem}\label{basiclemma}
Given a multidimensional Poisson bracket of the form (\ref{DNbracket}), for each $\alpha\,,$ the corresponding summand on the right-hand side of (\ref{DNbracket}) is a one-dimensional Poisson bracket of hydrodynamic type.
\end{lem}
\noindent Since one may perform linear transformations in the $X^\alpha$-variables it follows that any linear combination of such one-dimensional Hamiltonian structures must also be Hamiltonian and hence any pair of summands in (\ref{DNbracket}) defines a $(1+1)$-dimensional biHamiltonian structure.

\section{Linear metrics}

In this section we consider families of metrics $g^{ij\alpha}\,,\alpha=1\,,\ldots\,,D$ that are of the form
\begin{equation}
g^{ij\alpha} = \left( b^{ij\alpha}_k + b^{ji\alpha}_k\right) u^k + g^{ij\alpha}_o\,.
\label{metric}
\end{equation}
By Lemma \ref{basiclemma}, this, for each separate $\alpha$, must define a $(1+1)$-dimensional Poisson structure of hydrodynamic type, and the conditions for this were derived by Gelfand and Dorfman \cite{GD} and further studied by Balinskii and Novikov \cite{BN}. These are best described in terms of a product
\[
e^i \circa e^j = b^{ij\alpha}_k e^k
\]
and inner product
\[
\langle e^i , e^j \rangle_\alpha = g^{ij\alpha}_o
\]
and the required conditions become the statement that this product defines a Novikov algebra with a compatible quasi-Frobenius inner product:

\begin{defn}

\begin{itemize}

\item[(a)]A Novikov algebra is a vector space $\mathcal{A}$ equipped with a composition (called multiplication) $\circ:\mathcal{A}\times\mathcal{A} \rightarrow\mathcal{A}$ with the properties

\begin{eqnarray*}
a \circ(b \circ c) - b \circ (a \circ c) & = & (a \circ b) \circ c - (b \circ a) \circ c\,,\\
(a \circ b) \circ c & = & (a \circ c) \circ b
\end{eqnarray*}
for all $a\,,b\,,c\in\mathcal{A}\,.$

\item[(b)] A compatible, or quasi-Frobenius inner product, on the algebra $\mathcal{A}$ is a map $\langle ~,~\rangle: \mathcal{A} \times \mathcal{A} \rightarrow \mathbb{C}$ with the property
\[
\langle a\circ b, c \rangle = \langle a, c \circ b\rangle\
\]
for all $a\,,b\,,c\in\mathcal{A}\,.$

\end{itemize}

\end{defn}

To find the full set of conditions one just substitutes equation (\ref{metric}) into Theorem \ref{mokhov1} which yields, on writing the resulting conditions algebraically in terms of the multiplications defined by the $b^{ij\alpha}_k$ and the inner products defined by the $g^{ij\alpha}_o\,:$

\begin{eqnarray}
\sum_{(\alpha,\beta)} ( u \circb v) \circa w + w \circa(u \circb v) - (w \circa v) \circb u - u \circb ( w \circa v) & = & 0 \,, \label{A}\\
\sum_{(u,v,w)} ( u \circb v) \circa w + w \circa(u \circb v) - (w \circa v) \circb u - u \circb ( w \circa v) & = & 0 \,,\label{B} \\
\sum_{(\alpha,\beta)} ( u \circa v) \circb w - ( u \circa w) \circb v & = & 0 \,,\label{C} \\
(u \circb v) \circa w - u \circb ( v \circa w) - ( v \circa u) \circb w + v \circa (u \circb w) & = & 0 \,, \label{D} \\
\sum_{(\alpha,\beta)} \langle u , v \circb w \rangle_\alpha - \langle v, u \circa w \rangle_\beta & = & 0 \,,\label{E} \\
\sum_{(u,v,w)} \langle u , v \circb w \rangle_\alpha - \langle v, u \circa w \rangle_\beta & = & 0 \,, \label{F}
\end{eqnarray}
Note, if $\alpha=\beta$ then these conditions reduce to the definition of a Novikov algebra with a compatible quasi-Frobenius inner product, in accordance with Lemma \ref{basiclemma}. Conditions (\ref{A}) and (\ref{C}) are equivalent to the requirement that the product
\[
u \circ v = u \circa v + u \circb v
\]
defines a Novikov algebra for arbitrary $\alpha$ and $\beta\,,$ and a simple calculation shows that condition (\ref{B}) follows from conditions (\ref{A},\ref{C}) and (\ref{D})\,. Condition (\ref{E}) is equivalent to the requirement that inner product
\[
\langle u,v \rangle = \langle u,v \rangle_\alpha + \langle u,v \rangle_\beta
\]
is quasi-Frobenius with respect to the product $u \circ v = u \circa v + u \circb v\,.$

We thus obtain the definition
of the algebra $\mathcal{B}\,:$

\begin{defn}
The algebra $\mathcal{B}$ consists of the vector space $\mathbb{C}^N$ equipped with multiplications $\circa$ and inner-products $\langle~,~\rangle_\alpha\,,\alpha=1\,,\ldots\,, D$ such that:

\begin{itemize}

\item[(i)] the multiplication $\circa$ defines Novikov algebra with a compatible, quasi-Frobenius inner-product $\langle~,~\rangle_\alpha$ for all $\alpha=1\,,\ldots\,,D\,;$

\item[(ii)] the sum $ u \circ v = u \circa v + u \circb v$ defines a Novikov algebra with a compatible, quasi-Frobenius inner-product $\langle u,v \rangle = \langle u,v \rangle_\alpha + \langle u,v \rangle_\beta$ for all $\alpha\,,\beta=1\,,\ldots\,,D\,;$

\item[(iii)] the following compatibility condition holds:

\begin{itemize}

\item[(a)]
\[
(u \circb v) \circa w - u \circb (v \circa w) = ( v \circa u) \circb w - v \circa ( u \circb w)\,
\]
\item[(b)]
\[
\sum_{(u,v,w)} \langle u , v \circb w \rangle_\alpha - \langle v, u \circa w \rangle_\beta = 0
\]
\end{itemize}
for all $\alpha\,,\beta=1\,,\ldots\,,D$ and for all $u,v,w \in \mathbb{C}^N\,.$

\end{itemize}

\end{defn}

\noindent This algebra was defined implicitly in \cite{M} - here we just write it out explicitly. We will concentrate on the explicit construction of algebras $\mathcal{B}$ in the following two cases:

\medskip

{\noindent\underline{Class A:}} homogeneous structures, i.e. $g^{ij\alpha}_o=0$ for all $\alpha\,;$

\medskip

{\noindent\underline{Class B:}} one of the terms is constant and the remaining terms homogeneous, i.e.

\begin{eqnarray*}
g^{ij\alpha}_o & = & 0\,,\qquad \circa \neq 0\,, \quad \alpha=2\,,\ldots\,,D\,;\\
g^{ij,\alpha=1}_o & \neq & 0 \,, \quad \circequalsone = 0 \,,\\
\end{eqnarray*}

\noindent with, in both cases, no requirement on the non-degeneracy of the structures. Note though, if {\sl one} of the metrics, say $g^{ij,\alpha=1}$ is non-degenerate, and hence a flat metric, one may introduce flat-coordinates in which $\circequalsone=0\,,$ and ${\rm det}(g^{ij,\alpha=1}_o) \neq 0\,,$ and so the second case above is, if one of the individual structures is non-degenerate, a subclass of the first. In case (b) the conditions linking the multiplications and innner-products reduce to the following:
\[
\langle u \circb v,w \rangle_{\alpha=1} = \langle u, w \circb v \rangle_{\alpha=1}\,,
\]
i.e. the single inner product is quasi-Frobenius with respect to each of the multiplications, and
\[
\sum_{(u,v,w)} \langle u, v \circb w \rangle_{\alpha=1} = 0 \,,
\]
as derived in \cite{M,M2}.

\section{A construction}

Examples of Novikov algebras are easy to construct. In the original paper of Gelfand and Dorfman \cite{GD} the following Novikov product was found
\[
a \circ b = a \cdot \delta b\,.
\]
Here $\cdot$ is an associative, commutative product and $\delta$ a derivation on this algebra. This result was extended by \cite{Fil} to a product
\[
a \circ b = v \, a \cdot b + a \cdot \delta b
\]
where $v$ is a constant element of the underlying field. In \cite{Xu} this was further extended to a product
\[
a \circ b = v \cdot a \cdot b + a \cdot \delta b
\]
where $v$ is a constant element in the algebra. Note that if an inner product $\langle~,~\rangle$ is Frobenius with respect to the product $\cdot$ then it is automatically quasi-Frobenius with respect to
the above Novikov product.

Since $v$ is a constant and derivations form a vector space one can easily construct pencils of Novikov algebras,
\begin{eqnarray*}
a \circa b & = & v_\alpha \cdot a \cdot b + a \cdot \deltaa b \,, \\
a \circb b & = & v_\beta \cdot a \cdot b + a \cdot \deltab b \,, \\
\end{eqnarray*}
(here we concentrate on the case where $v$ is a constant vector rather than a scalar. A near identical result holds if one ignores the first product and regard
$v$ as a constant scalar). Note, the labels $\alpha\,,\beta$ on the derivation just denotes its pairing with the $v_\alpha\,,v_\beta$ in the first term of these products. Clearly
\begin{eqnarray*}
u \circ v & = & u \circa v + u \circb v \,, \\
& = & (v_\alpha + v_\beta) \cdot u \cdot v + u \cdot (\deltaa+\deltab) v  \\
\end{eqnarray*}
defines a Novikov algebra.

Note, without further restrictions on the constant vectors and derivations, this result provides a construction of pencils of flat homogeneous metrics, and hence examples of bi-Hamiltonian structures in $(1+1)$-dimensions. To find Hamiltonian structures in $(1+D)$-dimensions within this class of pencils of Novikov algebras places restrictions on the data $\{v_\alpha\,,\deltaa\}$.

\begin{prop}
Let $\mathcal{A}$ be a commutative, associative algebra with product $\cdot\,:\mathcal{A} \times \mathcal{A} \rightarrow \mathcal{A}$ and space of derivations
${\sl der}(\mathcal{A})\,.$ Consider the set $\{v_\alpha\,,\deltaa\,:\alpha=1\,, \ldots \,, M\}$ of constant vectors and derivations with the properties
\begin{eqnarray*}
[ \deltaa, \deltab ] & = & \alpha \cdot \deltab - \beta \cdot \deltaa \,, \\
\deltaa (v_\beta) - \deltab (v_\alpha) & = & 0 \,.
\end{eqnarray*}
Then the set $\{ \mathbb{C}^N\,; \circa\,:\alpha=1\,, \ldots \,, M\}$ is a $\mathcal{B}$-algebra where the products $\circa$ are defined by
\[
u \circa v = v_\alpha \cdot u \cdot v + u \cdot \deltaa v\,.
\]
This defines an $M$-dimensional, $N$-component Hamiltonian structure.

\end{prop}

The following results are immediate:

\begin{cor} Let $\mathcal{A}$ be an associative, commutative algebra with a two-dimensional space of derivations ${\sl der}(\mathcal{A})$. Then the above construction gives a $D=2\,,$ ${\sl dim}(\mathcal{A})$-component Hamiltonian structure.
\end{cor}

\begin{proof} Since derivations form a Lie algebra and ${\sl der}(\mathcal{A})$ is two dimensional, then automatically
\[
[\delta_1,\delta_2] = k_1 \delta_2 - k_1 \delta_1
\]
for some constants $k_1$ and $k_2\,.$ Hence one obtains compatible products
\begin{eqnarray*}
a \circone b & = & k_1 \, a \cdot b + a \cdot \delta_1 b \,, \\
a \circtwo b & = & k_2 \, a \cdot b + a \cdot \delta_2 b
\end{eqnarray*}
and hence an algebra $\mathcal{B}= \{ \mathbb{C}^N\,; \circone\,,\circtwo \}$ and a 2-dimensional, ${\sl dim}(\mathcal{A})$-component, Hamiltonian structure
\end{proof}

\begin{lem} Consider the space ${\widetilde{\sl der}}(\mathcal{A})\subset {\sl der}(\mathcal{A})$ of derivations with the properties
\begin{eqnarray*}
[ \deltaa, \deltab ] & = & \alpha \cdot \deltab - \beta \cdot \deltaa \,, \\
\deltaa (v_\beta) - \deltab (v_\alpha) & = & 0 \,
\end{eqnarray*}
for some set of vectors $\{ v_\alpha\,,i=1\,,\ldots\,,M\}.$
Then this is a Lie algebra with respect to the structure $[\deltaa,\deltab]=\deltaa\deltab - \deltab\deltaa\,.$
\end{lem}
\begin{proof} Purely computational
\end{proof}

\noindent We illustrate this construction with the following example.

\begin{example}\label{basicexample}
Consider the commutative, associative algebra
\[
\mathcal{A} \cong \mathbb{C}[z]/{\langle z^N \rangle}
\]
with basis $e^i = z^{i-1}\,,i=1\,,\ldots\,,N$ (and with the convention that $e^i=0$ for $i>N$). In this basis the multiplication is just
\[
e^i \cdot e^j = e^{i+j-1}\,.
\]
The derivations on this algebra are
\[
\delta_k e^i = (i-1) e^{i+k-1}\,, \qquad k=1\,,\ldots\,,N-1
\]
and these form the Lie algebra $[\delta_i,\delta_j]=(j-i) \delta_{i+j-1}\,.$ In particular,
\begin{eqnarray*}
[\delta_1,\delta_k] & = & (k-1) \delta_k\,,\\
& = & (k-1) \delta_k - 0 \, \delta_1\,.
\end{eqnarray*}
Thus, by the above Corollary, one obtains pairs of compatible products, namely
\begin{eqnarray*}
u \circone v & = & (k-1) u \cdot v + u \cdot \delta_1 v\,,\\
u \circk v & = & 0 \, u \cdot v + u \cdot \delta_k v\,,
\end{eqnarray*}
or, in the above basis, the products
\begin{eqnarray*}
e^i \circone e^j & = & (j+k-2) e^{i+j-1}\,,\\
e^i \circk e^j & = & (j-1) e^{i+j+k-2}\,.
\end{eqnarray*}
Thus the construction yields $(N-2)$ pairs of compatible multiplications, and hence $(N-2)$ examples of $2$-dimensional, $N$-component Hamiltonian structures. The corresponding metrics are:
\begin{eqnarray*}
\left. g^{ij\alpha}\right|_{\alpha=1} & = & \left( i+j-2(k-2) \right) u^{i+j-1}\,,\\
\left. g^{ij\alpha}\right|_{\alpha=k} & = & \left( i+j-2 \right) u^{i+j+k-2}\,
\end{eqnarray*}
where $k\in\{2\,,\ldots\,,N-2\}\,.$
The first of these metrics is always non-degenerate and so may be, via a change of variable, transformed to a constant form. For example, if $N=3$ (so $k=2$) one obtains the metrics
\begin{eqnarray*}
\left. g^{ij\alpha}\right|_{\alpha=1} & = &
\left(
\begin{array}{ccc}
2 u^1 & 3 u^2 & 4 u^3 \\
3 u^2 & 4 u^3 & 0 \\
4 u^3 & 0 & 0
\end{array}
\right)\,,\\
\left. g^{ij\alpha}\right|_{\alpha=k} & = &
\left(
\begin{array}{ccc}
0 & u^3 & 0 \\
u^3 & 0 & 0 \\
0 & 0 & 0
\end{array}
\right)\,.
\end{eqnarray*}
Solving the associated Gauss-Manin equations (for a systematic approach to solving these equations for this class of metrics, see \cite{Str}) gives the coordinate transformation
\[
u^1 = v^1 v^3 + \frac{1}{8} (v^2)^2\,, \qquad  u^2 = v^2 (v^3)^2\,, \qquad u^3 = (v^2)^4
\]
and in the new coordinates one obtains the pair
\begin{eqnarray*}
\left. g^{ij\alpha}\right|_{\alpha=1} & = &
\left(
\begin{array}{ccc}
0 & 0 & 1 \\
0 & 4 & 0 \\
1 & 0 & 0
\end{array}
\right)\,,\\
\left. g^{ij\alpha}\right|_{\alpha=k} & = &
\left(
\begin{array}{ccc}
-\frac{1}{2} v^2 & v^3 & 0 \\
v^3 & 0 & 0 \\
0 & 0 & 0
\end{array}
\right)\,.
\end{eqnarray*}

\end{example}

\noindent Thus this examples furnishes examples in both of the classes of structures defined above. We end this section with another example, this giving a $4$-component, $3$-dimensional example.

\begin{example}
Consider the commutative, associative algebra with multiplication table
\[
\begin{array}{c|cccc}
\cdot & e^1 & e^2 & e^3 & e^4 \\
\hline
e^1 & e^1 & e^2 & e^3 & e^4 \\
e^2 & e^2 & e^4 & e^4 & 0 \\
e^3 & e^3 & e^4 & 0 & 0 \\
e^4 & e^4 & 0 & 0 & 0
\end{array}
\]
The derivations of this algebra are easily constructed:
\[
\delta_1
\left(
\begin{array}{c}
e^1 \\
e^2 \\
e^3 \\
e^4
\end{array}
\right) =
\left(
\begin{array}{c}
0  \\
e^2 \\
e^3 \\
2 e^4
\end{array}
\right)
\,,\qquad
\delta_2
\left(
\begin{array}{c}
e^1 \\
e^2 \\
e^3 \\
e^4
\end{array}
\right) =
\left(
\begin{array}{c}
0  \\
e^3 \\
2 e^3 \\
2 e^4
\end{array}
\right)
\,,
\]
\[
\delta_3
\left(
\begin{array}{c}
e^1 \\
e^2 \\
e^3 \\
e^4
\end{array}
\right) =
\left(
\begin{array}{c}
0  \\
e^4 \\
0 \\
0
\end{array}
\right)
\,,\qquad\delta_4
\left(
\begin{array}{c}
e^1 \\
e^2 \\
e^3 \\
e^4
\end{array}
\right) =
\left(
\begin{array}{c}
0  \\
0 \\
0 \\
e^4
\end{array}
\right)
\,,
\]
\noindent and the set $\{\delta_1\,,\delta_3\,,\delta_4\}$ satisfy the conditions
\begin{eqnarray*}
{[}\delta_1,\delta_3{]} & = & 0 \,.\, \delta_1 - 1 \,.\, \delta_3\,,\\
{[}\delta_1,\delta_4{]} & = & 0 \,.\, \delta_1 - 1 \,.\, \delta_4\,,\\
{[}\delta_3,\delta_4{]} & = & 0 \,.\, \delta_3 - 0 \,.\, \delta_4\,,
\end{eqnarray*}
and hence the above construction may be implemented. This gives the metrics
\[
g^{ij,\alpha=1} =
\left(
\begin{array}{cccc}
2 u^1 & 3 u^2 & 3 u^3 & 4 u^4 \\
3 u^2 & 4 u^4 & 4 u^4 & 0 \\
3 u^3 & 4 u^4 & 0 & 0 \\
4 u^4 & 0 & 0 & 0
\end{array}
\right)
\]
\[
g^{ij,\alpha=3} =
\left(
\begin{array}{cccc}
0 & u^4 & 0 & 0 \\
u^4 & 0 & 0 & 0 \\
0 & 0 & 0 & 0 \\
0 & 0 & 0 & 0
\end{array}
\right)\,,\qquad
g^{ij,\alpha=4} =
\left(
\begin{array}{cccc}
0 & 0 & u^4 & 0 \\
0 & 0 & 0 & 0 \\
u^4 & 0 & 0 & 0 \\
0 & 0 & 0 & 0
\end{array}
\right)\,.
\]
Since the first metric is non-degenerate one may easily solve the Gauss-Manin equations to find its flat coordinates. Applying this transformation to the full set of structures gives
\[
g^{ij,\alpha=1} =
\left(
\begin{array}{cccc}
0 & 0 & 0 & 1 \\
0 & 4 & 4 & 0 \\
0 & 4 & 0 & 0 \\
1 & 0 & 0 & 0
\end{array}
\right)
\]
\[
g^{ij,\alpha=3} =
\left(
\begin{array}{cccc}
-\frac{1}{2} v^3 & v^4 & 0 & 0 \\
v^4 & 0 & 0 & 0 \\
0 & 0 & 0 & 0 \\
0 & 0 & 0 & 0
\end{array}
\right)\,,\qquad
g^{ij,\alpha=4} =
\left(
\begin{array}{cccc}
-\frac{1}{2}(v^3-v^2) & 0 & v^4 & 0 \\
0 & 0 & 0 & 0 \\
v^4 & 0 & 0 & 0 \\
0 & 0 & 0 & 0
\end{array}
\right)\,.
\]
\end{example}

\medskip

\noindent This example may be easily generalized.

\section{Monodromy and multi-dimensional Hamiltonian structures}

Throughout this section we assume that one of the metrics, say $g^{ij,\alpha=1}$, is non-degenerate, and that the associated Novikov algebra $\mathcal{A}$ has a right identity. The flat coordinates are found by solving
the Gauss-Manin equations
\[
{}^{g^{\alpha=1}} \nabla d{\bf v} = 0\,,
\]
where ${}^{g^{\alpha=1}} \nabla$ is the Levi-Civita connection corresponding to the metric $g^{\alpha=1}\,.$ This is an over-determined system, but the compatibility conditions are precisely the flatness of the metric and hence are automatically satisfied.

Solutions of the Gauss-Manin equations will exhibit branching around the discriminant
\[
\Sigma=\{ {\bf u}\, | \,\Delta({\bf u})=0 \}
\]
where $\Delta({\bf u}) = \det\left(g^{ij}({\bf u})\right),$ and this is encapsulated in the associated monodromy group
\[
W(M) = \mu \left( \pi_1(M\backslash\Sigma) \right)
\]
(see, for example, \cite{D}). Explicitly, the continuation of a solution under a closed path $\gamma$ on $M\backslash\Sigma$ yields a transformation
\[
{\tilde v}^a ({\bf u}) = A^a_b(\gamma) v^b({\bf u}) + B^a(\gamma)
\]
with $A$ orthogonal with respect to the metric $g({\bf u})\,,$ and these generate a subgroup of $O(N)\,.$ Thus to every Novikov algebra $\mathcal{A}$ there is an associated monodromy group which we
denote $\mathcal{W}(\mathcal{A})\,.$

\begin{example}\label{basicexamplecont} For the algebra $\mathcal{A}$ defined by the first metric in Example \ref{basicexample} the monodomy group is $\mathcal{W}({\mathcal{A}}) \cong \mathbb{Z}_4[1,2,3]\,,$ which acts on the ${\bf v}$-variables
via
\begin{eqnarray*}
v^1 & \mapsto & \varepsilon^3 v^1 \,, \\
v^2 & \mapsto & \varepsilon^2 v^2 \,, \\
v^3 & \mapsto & \varepsilon^1 v^3
\end{eqnarray*}
where $\varepsilon^4=1\,.$
\end{example}
This monodromy group also controls the structure of the remaining metrics $g^{ij\alpha}$ for $\alpha\geq 2\,.$ On applying the transformation ${\bf u} = {\bf u}( {\bf{v}})$ to the entire pencil of
metrics
\[
g^{ij}({\bf u}) = \sum_{\alpha=1}^D k_\alpha g^{ij\alpha}({\bf u})\,, \qquad k_\alpha \in \mathbb{C}
\]
one knows that this transforms to the pencil
\[
g^{ij}({\bf v}) = \eta^{ij} + \sum_{\alpha=2}^D k_\alpha d^{ij\alpha}_k v^k
\]
for some constants $d^{ij\alpha}_k$ since the linearity property must be preserved under this change of variables. On applying the monodromy transformation to this pencil gives the follows invariance properties
for the component parts of the pencil:
\[
A^i_a \eta^{ab} A^j_b = \eta^{ij}
\]
and
\[
A^i_a d^{ab\alpha}_k A^j_b = d^{ab\alpha}_k
\]
and hence (assuming that the monodromy group acts such that $A^i_a = {\rm diag}( \varepsilon^{d_1}\,, \ldots \,, \varepsilon^{d_M})$ where $\varepsilon$ is the $(d_1+d_M)^{\rm th}$ root of unity) the constraints
\begin{eqnarray*}
\varepsilon^{d_a+d_b} \eta^{ab} & = & \eta^{ab} \,,\\
\varepsilon^{d_a + d_b - d_c} d^{ab\alpha}_c & = & d^{ab\alpha}_c
\end{eqnarray*}
and this constrains where the non-zero terms can appear: if
\begin{eqnarray*}
\varepsilon^{d_a+d_b} & \neq 0 & {\rm then~} \eta^{ab}=0\,,\\
\varepsilon^{d_a+d_b-d_c} & \neq 0 & {\rm then~} d^{ab\alpha}_c=0\,.
\end{eqnarray*}
This, coupled to the triangular structure of the transformation ${\bf u} = {\bf u}( {\bf{v}})$  which implies the the metrics $g^{ij\alpha}({\bf v})$ for $\alpha\geq 2$ must have entries on or above the antidiagonal, 
constrains these constants \cite{Str}.

\begin{example} In the above Example \ref{basicexamplecont}, the monodromy group $\mathcal{W} ( \mathcal{A} ) \cong \mathbb{Z}_4[1,2,3]$ and the only non-zero entries are $d^{11}_2\,,d^{12}_3$ and hence the pencil of metric takes the 
skeletal form
\[
g^{ij} = \left(\begin{array}{ccc} 0 & 0 & c_1 \\ 0 & c_2 & 0 \\ c_3 & 0 & 0 \end{array}\right) + \left(\begin{array}{ccc} d^{11}_2 v^2 & d^{12}_3 v^3 & 0 \\ d^{12}_3 v^3 & 0 & 0 \\ 0 & 0 & 0 \end{array}\right)\,.
\]
\end{example}
This methods only fixes the skeletal forms of the pencil: it does note fix the values of the constants. These have to be found by ensuring the resulting structures satisfy the $\mathcal{B}$-algebra conditions.

\section{Conclusion}
The above construction could be generalized: underlying the construction is a single commutative, associative algebra. One could start with pencils of compatible commutative, associative algebras and construct their derivations. One could then derive conditions for the resulting structures to satisfy the conditions in the definition og the $\mathcal{B}$-algebra. Another direction of possible research is to connect this construction with the theory of Segre forms, as used in \cite{FLS}. Conversely, normal forms has played no role in the theory of Novikov algebras, which has been more algebraic than geometric: it would be of interest to reconcile these two approaches.

Finally, it should be stressed that the construction presented here is far from exhaustive: it does not, in any way, claim to be part of a classification scheme - just a method to construct examples. But in an area where few examples are known, it does provide a method for the construction of multi-component, multi-dimensional brackets of Dubrovin-Novikov type.

\section*{Acknowledgements}
I would also like to thank  Blazej Szablikowski and Dafeng Zuo for various conversations concerning Novikov algebras and their role in the theory of Hamiltonian structures.

\end{document}